\documentclass[superscriptaddress,pra,twocolumn]{revtex4-1}
\usepackage{amsmath}
\usepackage{amsthm}
\usepackage{amsfonts}
\usepackage{graphicx}
\usepackage{bbold}
\usepackage{stmaryrd}
\usepackage{color}
\usepackage{hyperref}

\newtheorem{prop}{Proposition}
\newtheorem{theo}{Theorem}

\begin{document}

\title{Algorithmic quantum simulation of memory effects}

\date{\today}

\author{U. Alvarez-Rodriguez} \email{unaialvarezr@gmail.com}
\affiliation{Department of Physical Chemistry, University of the Basque Country UPV/EHU, Apartado 644, 48080 Bilbao, Spain} 
\author{R. Di Candia}
\affiliation{Department of Physical Chemistry, University of the Basque Country UPV/EHU, Apartado 644, 48080 Bilbao, Spain}
\affiliation{Dahlem Center for Complex Quantum Systems, Freie Universit\"at Berlin, 14195 Berlin, Germany}
\author{J. Casanova}
\affiliation{Institut f\"{u}r Theoretische Physik, Albert-Einstein-Allee 11, Universit\"{a}t Ulm, D-89069 Ulm, Germany}
\author{M. Sanz}\email{mikel.sanz@ehu.eus}
\affiliation{Department of Physical Chemistry, University of the Basque Country UPV/EHU, Apartado 644, 48080 Bilbao, Spain}
\author{E. Solano}
\affiliation{Department of Physical Chemistry, University of the Basque Country UPV/EHU, Apartado 644, 48080 Bilbao, Spain}
\affiliation{IKERBASQUE, Basque Foundation for Science, Maria Diaz de Haro 3, 48013 Bilbao, Spain}

\begin{abstract}
We propose a method for the algorithmic quantum simulation of memory effects described by integrodifferential evolution equations. It consists in the systematic use of perturbation theory techniques and a Markovian quantum simulator. Our method aims to efficiently simulate both completely positive and nonpositive dynamics without the requirement of engineering non-Markovian environments. Finally, we find that small error bounds can be reached with polynomially scaling resources, evaluated as the time required for the simulation.
\end{abstract}

\maketitle

\textit{Introduction.} Fundamental interactions in nature are described by mathematical models that frequently overcome our analytical and numerical capacities. This problem is especially challenging in the quantum realm, due to the exponential growth of the Hilbert space with the number of particles involved. Richard Feynman proposed~\cite{Feynman82} that the desired calculations may be experimentally realized by codifying the model of interest into the degrees of freedom of another more controllable quantum system. Along these lines, in the last decade, this approach has been employed to simulate the dynamics of many-body quantum systems. A machine performing this task is called quantum simulator, and it has been studied with increasing interest, theoretically and experimentally, in controlled quantum systems~\cite{trotzky, fm}. It is expected that quantum simulators will solve relevant problems unreachable for classical computers. Among them, we could mention complex spin, bosonic, and fermionic many-body systems~\cite{ss, fm}, entanglement dynamics~\cite{emb, embexp, embi}, and fluid dynamics~\cite{flu}, among others.

In quantum mechanics, realistic situations in which the quantum system is coupled to an environment are modeled in the framework of open quantum systems. In this description, an effective evolution equation for the system of interest is obtained by disregarding the environmental degrees of freedom~\cite{brepe}. The resulting dynamics can be classified as Markovian or non-Markovian \cite{rhp,blpv,mnm,car,dino,nmframe}. In the former, the time evolution depends solely on the current state of the system, and there are several results concerning its quantum simulation~\cite{bac,swe,rob, eis}. On the contrary, the non-Markovian evolution depends on the history of the system, and it is more challenging to treat both analytically and numerically~\cite{eis}. In this sense, despite some recent results~\cite{chiu,div,qyr,xian,col,sil,ros,petr,jiji,mli,mef}, including a work on the sufficient conditions for a completely positive and trace preserving (CPTP) non-Markovian dynamics \cite{suff}, a general non-Markovian quantum simulator has not been fully developed yet. A paradigmatic feature of non-Markovian dynamics is the existence of quantum memory effects as an extension of the classical history-dependent dynamics to the quantum domain. Moreover, a number of key applications in the quantum domain can be envisioned, such as quantum machine learning ~\cite{qml1,qml2}, neuromorphic quantum computing ~\cite{qmem,cmem} and quantum artificial life~\cite{qalife,bioc}. These can be implemented by mirroring the already existing results in memcomputing devices~\cite{mem}, intelligent materials~\cite{mm} and population dynamics~\cite{lv}. Therefore, the simulation of quantum memory effects would be a significant step forward to understand open quantum systems and, consequently, to employ them in the development of the aforementioned research fields.  

In this Rapid Communication, we provide an efficient and general framework for an algorithmic quantum simulation~\cite{aqst} of memory effects modeled by integrodifferential evolution equations. The protocol algorithmically combines a Markovian quantum simulator with perturbation theory techniques in order to retrieve the time evolution of an arbitrary initial state. Our method does not require the engineering of any additional environment, avoiding the challenging task of developing first-principle non-Markovian quantum simulators. Moreover, the protocol works even when the evolution does not correspond to a CPTP map, which is the case of most of time-delayed Lindblad master equations. Indeed, although the CPTP character is not guaranteed, our approach circumvents this issue by splitting the simulation into two CPTP parts. Finally, we prove polynomial scaling error bounds for the proposed method.

\textit{Integrodifferential equations with memory.} The model describing the memory effects we aim to simulate is based on the integrodifferential equation
\begin{equation}\label{nonmark}
\partial_t \rho(t)= \int_0^tds\, K(t,s)\mathcal{L}\,\rho(s) .
\end{equation}
Here, $K(t,s)$ is a memory Kernel modeling how the evolution of the state at a certain time is affected by its history, and $\mathcal{L}$ is a general time-independent Lindblad operator. Notice that $K(t,s)=2\delta (t-s)$ corresponds to the standard Markovian master equation written in the Lindblad form. As noticed, for instance, in Refs.~\cite{eis,qyr}, it is not conceivable to simulate a general non-Markovian dynamics efficiently. The reason is that one could then imagine simulating a highly inefficient calculation in the environment, retrieving this information afterwards into the system due to the non-Markovian information backflow, in an efficient manner. However, Eq.~\eqref{nonmark} includes in the kernel $K(t,s)$ the non-Markovian aspects of the evolution, which gives only an effective description of the environment contribution.

In order to simulate Eq.~\eqref{nonmark}, we use as a tool the quantum simulation of the equation
\begin{equation}\label{semimark}
\partial_t \rho(t)=\int_0^tds\, H(t,s)\left[\mathcal{E}-\mathcal{I}\right]\rho(s),
\end{equation}
where $H(t,s)$ is a memory kernel, $\mathcal{E}$ is a general CPTP map and $\mathcal{I}$ is the identity map. Equation~\eqref{semimark} describes the dynamics of a {\it semi-Markovian} process \cite{sem}. It is noteworthy to mention that while Eqs.~\eqref{nonmark} and \eqref{semimark} preserve the trace of the density matrix, they do not generally preserve positivity. However, sufficient conditions for Eq.~\eqref{semimark} to determine a CPTP map have been studied when $H(t,s)= H(t-s)$. Indeed, if the Laplace transform of the memory kernel $H(\tau)$ satisfies the relation $\tilde{H}(u)=\frac{u \tilde{w} (u)}{1-\tilde{w}(u)}$ for some waiting distribution $w(t)$, then Eq.~\eqref{semimark} corresponds to a CPTP process~\cite{bou}. Moreover, if this condition is fulfilled, then the solution of Eq.~\eqref{semimark} can be written as $\rho(t)=\sum_{i=0}^\infty p_i(t)\mathcal{E}^i\rho(0)$, where $0\leq p_i(t)\leq 1$~\cite{bou}. In this case, by truncating the series, we can simulate Eq.~\eqref{semimark} assuming that an efficient quantum simulator of $\mathcal{E}$ and its powers is available. In the following, we will consider processes $\mathcal{E}$ corresponding to Markovian evolutions, whose efficient quantum simulator has been already designed, e.g., $k$-local Lindblad equations~\cite{eis}. We will show how to simulate a general kernel $H(t,s)$, including the case in which Eq.~\eqref{semimark} does not correspond to a CPTP process. Finally, we illustrate how to employ this result to simulate Eq.~\eqref{nonmark}.

\textit{Algorithmic quantum simulator.} Let us consider the Volterra version of Eq.~\eqref{semimark},
\begin{equation}\label{ori}
\rho(t)=\rho(0)+\int^{t}_{0} ds \, h(t,s) \left[\mathcal{E}-\mathcal{I}\right] \rho (s),
\end{equation}
where $h(t,s)\equiv \int_s^td\tau\, H(\tau, s)$. We assume that $H(t,s)\geq 0$ and $h(t,s)\leq c$, for a given constant $c$, for all $t\geq s \geq 0$. Moreover, we quantify the results in terms of the trace norm for matrices, defined as the sum of their singular values $\| \sigma\|_1\equiv\sum_i \sigma_i$, and the respective induced superoperator norm $\left\| \mathcal{A}\right\|\equiv \max_{ \sigma}\frac{\left \| \mathcal{A} \sigma\right\|_1}{\| \sigma\|_1}$. Then, Eq. ~\eqref{ori} can be solved iteratively, via the series  $\rho(t)=\sum^{\infty}_{i=0} \rho_i(t)$, where 
\begin{equation}
\rho_0 = \rho(0), \quad \rho_{i \geq 1} = \int^{t}_{0}ds \, h(t,s) \left[\mathcal{E}-\mathcal{I}\right] \rho_{i-1}(s) .
\label{cpe}
\end{equation} 
This expansion can be truncated at order $n$, $\tilde{\rho}_n(t)=\sum_{i=0}^n \rho_i(t)$, with a small truncation error given by the following estimation.
\begin{prop}[Truncation error] \label{trunc}
$\| \rho(t) - \tilde \rho_M (t)\|_1 \leq \varepsilon$ provided that $M\geq at +\log1/\varepsilon-1$, with $a=(e+1)c\left\| \mathcal{E}-\mathcal{I}\right\|$.
\end{prop}

The proof of Proposition 1 is provided in Appendix~A. This truncation allows us to write the approximated solution of Eq.~\eqref{semimark} by a finite sum, with a number of terms growing linearly with the simulated time. Indeed, we have that
\begin{equation}\label{lin} 
\tilde \rho_n (t) = \sum_{i=0}^n d_i (t) \left[\mathcal{E}-\mathcal{I}\right]^i\rho(0),
\end{equation}
with the corresponding parameter values $d_0 (t) = 1$ and $d_{i\geq 1}(t)=\int_0^{s_0\equiv t}\cdots\int_0^{s_{i-1}}ds_1\cdots ds_{i}\,h(t,s_1)\cdots h(s_{i-1},s_i)$. This truncated sum can be rewritten as $\tilde \rho_n(t)=\sum_{i=0}^n c_i(t)\mathcal{E}^i\rho(0)$, with $c_i(t)=\sum_{k=i}^n {k \choose i}(-1)^{k-i}d_k(t)$. Proposition~\ref{trunc} tells us that we can directly simulate a semi-Markovian dynamics by just implementing powers of the process $\mathcal{E}$, and numerically integrating the memory kernel. As we need a number of terms which increases linearly with the simulated time, we have that this step is efficient if the implementation of the $\mathcal{E}$ is efficient. Notice that, by construction, $\tilde \rho_n$ has trace $1$, but it is not necessarily a density matrix, since it can have negative eigenvalues. However, we can write $\tilde \rho_n(t)$ as a weighted sum of two density matrices and introduce the quantities $c^+_i(t)\equiv \max\{c_i(t),0\}$ and $c^-_i(t)\equiv\min\{c_i(t),0\}$. In consequence, we have that  
\begin{equation}\label{solsem}
\tilde \rho_n (t) = C^+_n(t) \rho^+_n(t) + C^-_n (t) \rho^-_n(t),
\end{equation}
where the parameter values $C^{\pm}_n(t)=\sum_{i=0}^nc^{\pm}_i(t)$ and $\rho^{\pm}_n(t)=\frac{1}{C^{\pm}_n(t)}\sum_{i=0}^n c^{\pm}_i(t)\mathcal{E}^i\rho(0)$, while $C_n^-(t)=1-C_n^+(t)$ holds due to trace preservation. Notice that $\rho^{\pm}_n(t)$ are two density matrices, as their trace is $1$ and they are, by construction, positive. Indeed, we have approximated the dynamics, denoted by $\Lambda(t)$ [$\rho(t)=\Lambda(t)\rho(0)$], corresponding to Eq.~\eqref{semimark}, as a weighted sum of two CPTP maps: $\Lambda(t)\simeq \Lambda_n (t) =  C^+_n(t) \Lambda^+_n(t)+C^-_n(t)\Lambda^-_n(t)$, with $\Lambda^{\pm}_n(t)=\frac{1}{C^{\pm}_n(t)}\sum_{i=0}^n c^{\pm}_i(t)\mathcal{E}^i$. The form of the resulting CPTP maps allows us to simulate Eq.~\eqref{semimark} by making use of a Markovian quantum simulator and numerical techniques. In fact, all $c_i(t)$, and thus also $c^{\pm}_i(t)$, can be classically computed, and the states $\rho^{\pm}_n(t)$ can be prepared assuming that the Markovian operations $\mathcal{E}^i$ ($0\leq i\leq n$) are available.
\begin{prop}[Simulation of semi-Markovian processes]\label{semim}
Let us consider the simulating dynamics $\Lambda_{sim} (t)=C^+_M(t) \tilde\Lambda^+_M(t)+C^-_M(t)\tilde \Lambda^-_M(t)$, where $\tilde \Lambda^{\pm}_M(t)=\frac{1}{C^{\pm}_M(t)}\sum_{i=0}^M c^{\pm}_i(t) \mathcal{\tilde E}^i$, $\mathcal{\tilde E}$ denotes an efficient quantum simulation of $\mathcal{E}$, and $M\geq at+\log 1/\tilde \varepsilon$. If $\| \mathcal{E}^i-\mathcal{\tilde E}^i\|\leq \delta$ requires a simulation time $\bar t= O\left({\rm poly}(i, 1/\delta)\right)$, then we can simulate the semi-Markovian process in Eq.~\eqref{semimark} within an error $\|\Lambda(t)-\Lambda_{sim}(t)\|_1\leq \tilde \varepsilon$ by using a simulation time $\tilde t~=~O\left( {\rm poly}(t, C^+_M(t)/\tilde \varepsilon)\right)$. 
\end{prop}
\begin{proof}
We have that $\| \Lambda(t) - \Lambda_{sim}(t)\| \leq \|\Lambda(t)- \Lambda_M(t)\| + \|  \Lambda_M(t)-\Lambda_{sim}(t)\|$. The first term is bounded by $\tilde \varepsilon/2$, as $M \geq at+\log 1/\tilde \varepsilon$. The second term can be bounded by $\| \Lambda_M(t)-\Lambda_{sim} (t)\|\leq C^+_M(t)\| \Lambda^+_M(t)-\tilde \Lambda^+_{M}(t)\| -C^-_M(t)\| \Lambda^-_M(t)- \tilde \Lambda^-_{M}(t)\|$. We have that $\|\Lambda^{\pm}_M(t)-\tilde \Lambda^{\pm}_M(t) \|\leq \tilde \varepsilon/4|C^\pm_M(t)|$, assuming $\| \mathcal{E}^i-\mathcal{\tilde E}^i\|\leq \tilde \varepsilon/4|C^\pm_M(t)|$. This requires a simulation time $\bar t= O({\rm poly}(t, C^+_M(t)/\tilde \varepsilon)$, where we have used that $C^-_M(t)=1-C^+_M(t)$.
\end{proof} 
Proposition~\ref{semim} allows us to compute approximately the evolution of expectation values of observables under the dynamics of Eq.~\eqref{semimark}. It is noteworthy to mention that our method does not require the engineering of any bath corresponding to a semi-Markovian dynamics. Instead, we have written the formal solution of Eq.~\eqref{semimark}, and exploit the availability of a Markovian quantum simulator generating $\mathcal{E}$ and its powers. This is possible due to the fast convergence of the exponential series, which limits the number of terms to be classically computed. Moreover, the truncation provided in Proposition~\ref{trunc} implies also that an efficient Markovian simulation is sufficient to approximatively generate the solution of Eq.~\eqref{semimark}.
 
While in the CPTP semi-Markovian case we can directly sample from the probability distribution of a given observable, since we are directly implementing the solution, for more general non-Markovian equations we only have access to expectation values, as this time the process is split into two parts. A consequent question is whether we can compute interesting quantities beyond mere observables with our algorithmic quantum simulator.  In the following, we study the example of the two-time correlation function of unitary operators, i.e. $D^{\Lambda(t)}_{\rho,U}={\rm Tr}[U(t)U(0) \rho]$. In the last expression, $U(t)\equiv \Lambda^*(t) U$, where  $\Lambda^*(t)$ is the dual of $\Lambda(t)$, defined as ${\rm Tr}[A(\Lambda (t)\cdot\sigma)]\equiv {\rm Tr}[(\Lambda^*(t)\cdot A)\sigma]$ for arbitrary $A$ and $\sigma$. First, let us notice that $D_{\rho,U}^{\Lambda(t)}\simeq C^+_n(t) D^{\Lambda^+_n(t)}_{\rho,U}+C^-_n(t)D^{\Lambda^-_n(t)}_{\rho,U}$ for a sufficiently large $n$. Each resulting term can be computed with an extension to unitary dynamics of the protocol for the two-time correlation function proposed in Ref.~\cite{Julen}. Indeed, we add a two-dimensional ancilla and initialize the joint system in the state $\tilde \rho = \frac{1}{2}(|0\rangle+|1\rangle)(\langle 0| + \langle 1|)\otimes \rho$. First, we implement a controlled operation $U_c=|0\rangle\langle 0| \otimes U + |1\rangle\langle 1|\otimes \mathbb{1}$, then the evolution $\Lambda^{\pm}_n (t)$ on the original system, and finally $U_c$ again. In the end, $D^{\Lambda^{\pm}_n(t)}_{\rho,U}$ is retrieved by measuring the operator $\langle\sigma_x\rangle+i\langle\sigma_y\rangle$ in the ancilla. Notice that this protocol shows the same efficiency as the one in Proposition~\ref{semim}. Moreover, the method can be straightforwardly extended to multi-time correlation functions of unitary operators by iterating the aforementioned steps. Lastly, the multi-time correlation functions of observables $O$ can be computed by decomposing it into $O= U_{a}+ \gamma U_b$, with $\gamma \in \mathbb{R}$ and $U_{a,b}$ unitary matrices (see Appendix~C and Ref.~\onlinecite{flu}).  

Now, we are ready to show how to use the quantum simulation of Eq.~\eqref{semimark} to simulate Eq.~\eqref{nonmark}. Let us consider $\mathcal{E}=e^{\lambda\mathcal{L}}$, where $\lambda\in \mathbb{R}^+$ is a control parameter and $\mathcal{L}$ an arbitrary Lindblad operator, as in Eq.~\eqref{nonmark}. In the following, we prove that the solution of Eq.~\eqref{semimark}, describing a semi-Markovian process, approximates the solution of the memory process in Eq.~\eqref{nonmark} provided that $\lambda$ is small. 

\begin{figure}[h!]
\includegraphics[width=0.5\textwidth]{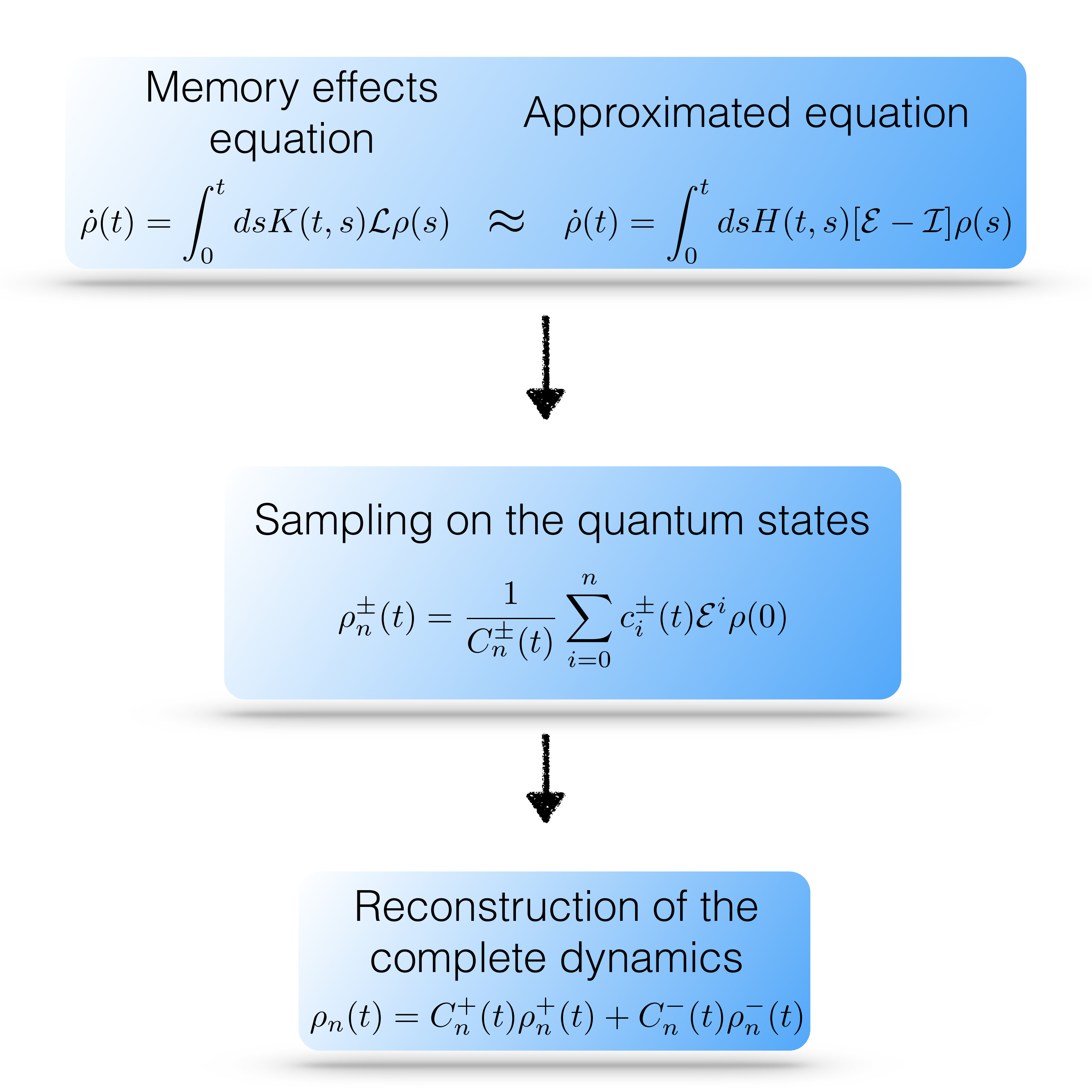}
\caption{ Scheme of our algorithmic quantum simulator. We approximate the equation underlining the memory effects with a semi-Markovian equation. We then split the solution of the semi-Markovian process into two CPTP parts, implementing each part separately. This process is accompanied by the integration of products of the memory kernel in a number which increases linearly with the simulated time.}  
\label{aqs}
\end{figure}

\begin{prop}[Simulation of memory effects]\label{prop3}
Let $\rho_1(t)$ and $\rho_2(t)$ be the solutions of Eq.~\eqref{nonmark} and Eq.~\eqref{semimark} respectively, with $\mathcal{E}\equiv \mathcal{E}_\lambda=e^{\lambda \mathcal{L}}$ ($\lambda\in \mathbb{R}^+$), $H(t,s)=K(t,s)/\lambda$  with $\int_s^t d\tau\, K(\tau,s)\leq c$, and $\rho_1(0)=\rho_2(0)$. Then, $\| \rho_1(t)-\rho_2(t)\|_1\leq \varepsilon$ holds if $\lambda \leq \frac{e^{-(2+o(\varepsilon))c\|\mathcal{L}\| t}}{c\left\| \mathcal{L}\right\|^2t}\varepsilon$, when $c\|\mathcal{L}\| t > 1/e$, and  if $\lambda \leq \log\left(\frac{1}{c \left\| \mathcal{L}\right\| t}\right)\frac{\varepsilon}{\|\mathcal{L}\|}$, when $c\|\mathcal{L}\| t \leq 1/e$, provided that $\varepsilon \le 1/2$. \end{prop}

The bounds of Proposition~\ref{prop3} are rigorously found in Appendix~B. The result of Proposition~\ref{prop3} provides the error bound for a general simulation of a complex environment described by Eq.~\eqref{nonmark}, and it is rather general as it holds for any $\mathcal{L}$. The algorithm consists in implementing the states defining the solution of the approximated semi-Markovian process, together with the numerical integration of the memory kernel, as schematically depicted in Fig.~\ref{aqs}. The method can be generalized to even more complicated dynamical equations. For instance, the case of higher-order derivatives, as
\begin{equation}\label{general}
\partial_t \rho (t)=\int^{t}_{0} \int^{s_1}_{0} ds_2\,ds_1\, K(s_1,s_2) \mathcal{L}\, \rho (s_2).
\end{equation}
The solution of Eq.~\eqref{general} can be approximated analogously to Eq.~\eqref{nonmark}, and Proposition~\ref{prop3} extended in order to find similar bounds. A further generalization consists in introducing additional terms, increasing the versatility of the proposed algorithmic quantum simulator. For instance, let us consider the equation
\begin{equation}\label{general1}
\partial_t \rho(t)=\sigma + \int^{t}_{0}ds\, K(t,s) \mathcal{L}\, \rho (s),
\end{equation}
where $\sigma$ can be an arbitrary matrix. Then, Eq.~\eqref{general1} can be simulated by approximating it with the equation $\partial_t \rho(t)=\sigma + \int^{t}_{0}ds\, K(t,s) \left[e^{\lambda \mathcal{L}} -\mathcal{I}\right]/\lambda\, \rho (s)$, which can be rewritten and simulated similarly to Eq.~\eqref{semimark}.

\textit{Conclusions.} We have developed a flexible and efficient quantum algorithm for the solution of integrodifferential evolution equations describing quantum memory effects, including the case of non-Markovian dynamics. The proposed algorithmic quantum simulation is useful for mimicking the effective action of complex environments. Alternative situations that our approach may cover include quantum feedback, quantum machine learning, and neuromorphic quantum computation. Lastly, the results in this Rapid Communication can be exploited for the classical simulation of memory effect equations. In fact, if the Markovian process used as a tool is decomposed efficiently by gates with a positive Wigner function, then expected values of observables can be estimated by using Monte Carlo techniques~\cite{qp1,qp2}.

{\it Acknowledgments.} The authors thank Ryan Sweke for useful discussions. The authors acknowledge support from Spanish MINECO/FEDER FIS2015-69983-P, UPV/EHU UFI 11/55 and a Postdoctoral fellowship, Basque Government Grant IT986-16 and BFI-2012-322, the Alexander von Humboldt Foundation and the AQuS EU project.


\appendix

\section{Proof of Proposition 1}
In this section, we provide in more detail the proof of Proposition~1, which gives us an upper bound for the truncation error. In order to quantify the error, we use the trace norm of a matrix, defined as the sum of the singular values of the matrix: $\|\sigma\|_1\equiv \sum_i\sigma_i$. The following recursion relation holds:
\begin{eqnarray}
\| \rho(t)-\tilde \rho_n(t)\|_1 \leq & \int^{t}_{0}ds\, h(t,s) \|\mathcal{E-I}\| \| \rho(t)-\tilde \rho_{n-1}(t)\|_1\nonumber \\ 
\leq & y c \int_0^t ds\, \|\rho(s)-\tilde \rho_{n-1}(s)\|_1 ,
\end{eqnarray}
where $\rho(t)$ is the ideal solution at time $t$, $\tilde \rho_n (t)$ is $n$th order truncation at time $t$, $h(t,s)\leq c$, and $y\equiv \left\|\mathcal{E}-\mathcal{I}\right\|$, in which the superoperator norm is induced by the trace norm, i.e. $\|\mathcal{A}\|\equiv \textrm{sup}_{\sigma}\frac{\|\mathcal{A}\sigma\|_1}{\|\sigma\|_1}$. The truncation error can be thus evaluated by induction, by considering the zeroth-order truncation error,
\begin{equation}\label{0th}
\|\rho(t)-\tilde \rho_0(t)\|_1\leq  y\int^{t}_{0}ds\, h(t,s) \|\rho(s)\|_1.
\end{equation}
A bound on $\|\rho(s)\|_1$ can be found by using a Gr\"{o}nwall's inequality.
\begin{theo}[Gr\"onwall's inequality~\cite{ineq}]\label{gro1} Let $u$ be a continuous function defined on $J=[\alpha,\beta]$ and let the function $g(t,s)$ be continuous and nonnegative on the triangle $\Delta:\alpha\leq s\leq t\leq \beta$ and nondecreasing in $t$ for each $s\in J$. Let $n(t)$ be a positive continuous and nondecreasing function for $t\in J$. If 
\begin{equation}
u(t) \le n(t) + \int^{t}_{\alpha} ds\, g(t,s) u(s),\qquad t\in J,
\end{equation}
then
\begin{equation}
u(t)\le n(t) e^{  \int^{t}_{\alpha}ds\, g(t,s)},\qquad t \in J.
\end{equation}
\end{theo}
One can prove from the Volterra equation that $\|\rho(t)\|_1\leq \|\rho(0)\|_1+y\int^{t}_{0}ds\, h(t,s)\|\rho(s)\|_1$. Theorem~\ref{gro1} implies that
\begin{equation}\label{b1}
\|\rho(t)\|_1 \le e^{y\int^{t}_{0}ds\, h(t,s)},
\end{equation}
where we have set $\|\rho(0)\|_1=1$. Here, we have assumed that $H(t,s)\geq0$, to satisfy the hypothesis on $h(t,s)=\int_s^td\tau\, H(\tau,s)$, in order to apply Theorem~\ref{gro1}. By plugging Eq.~\eqref{b1} into Eq.~\eqref{0th}, we find that
\begin{equation}\label{in2}
\|\rho(t)-\tilde \rho_0(t)\|_1\leq y \int_0^t ds\, h(t,s)e^{y\int_0^sd\tau\, h(s,\tau)}\leq e^{cyt}-1,
\end{equation}
where in the second inequality, we have used that $h(s,\tau)\leq h(t,\tau)$ for $s\leq t$, allowing us to perform the integration. In the third inequality we have made use of the bound on $h(t,s)$, $h(t,s)\leq c$.

We can now prove by induction that 
\begin{equation}\label{ind}
\|\rho(t)-\tilde \rho_M(t)\|_1\leq \sum_{i=M+1}^\infty \frac{(cyt)^i}{i!},
\end{equation}
for any natural $M$. The case $M=0$ is just the inequality found in Eq.~\eqref{in2}. Let us assume that Eq.~\eqref{ind} holds for $M=n-1$. Then,
\begin{align*}
\|\rho(t)-\tilde \rho_{n}(t)\|_1\leq & y c \int_0^tds\,\| \rho(s)-\tilde \rho_{n-1}(s)\|_1\\ 
\leq & y c \int_0^tds\, \sum_{i=n}^\infty \frac{(cys)^i}{i!}=\sum_{i=n+1}^\infty \frac{(cyt)^i}{i!},
\end{align*} 
which concludes the proof of Eq.~\eqref{ind}.

In the following, we will prove that $\sum_{i=M+1}^\infty x^i/i!\leq \varepsilon$ holds, provided that $M\geq (e+1)x +\log(1/\varepsilon)-1$. Indeed, we have that
\begin{eqnarray*}
&&\sum_{i=M+1}^\infty \frac{x^i}{i!} \leq e^x\frac{x^{M+1}}{(M+1)!} \leq e^x\left(\frac{ex}{M+1}\right)^{M+1} \\ 
&&=e^x\left(1+\frac{ex-(M+1)}{M+1}\right)^{M+1} \leq e^xe^{ex-(M+1)} \leq \varepsilon.
\end{eqnarray*}
In the first inequality, we have used the Lagrange error formula for the Taylor expansion of the exponential series. In the second inequality, we have used the Stirling inequality $n! \geq \left( \frac{n}{e}\right)^n$. In the third inequality, we have used that $ \left(1+\frac{a}{b}\right)^b \le e^a$. Finally, in the last inequality of Eq.~(9), we have used the lower bound on $M$. By applying the last result to $x=cyt$, we finish the proof of Proposition~1.

\section{Proof of Proposition $3$}

In Proposition 3, we estimate the error made when approximating the equation 
\begin{equation}\label{eq1}
\partial_t \rho(t) = \int_0^t ds\, K(t,s)\mathcal{L}\rho(s)
\end{equation}
by the equation which corresponds to the semi-Markovian process
\begin{equation}\label{eq2}
\partial_t \rho(t) = \int_0^t ds\, H(t,s)\left[\mathcal{E_\lambda}-\mathcal{I}\right]\rho(s),
\end{equation}
where $H(t,s)= K(t,s)/\lambda$ and $\mathcal{E_\lambda}=e^{\lambda\mathcal{L}}$, with the same initial condition for both equations. Let us denote by $\rho_1(t)$ and $\rho_2(t)$ the solutions to Eqs. \eqref{eq1} and \eqref{eq2}, respectively. Considering the corresponding Volterra equations, we can upper bound the distance between $\rho_1(t)$ and $\rho_2(t)$,
\begin{eqnarray}
&& \nonumber \|\rho_1(t)-\rho_2(t)\|_1 \\ && \nonumber = \left\| \int^{t}_{0}ds\, k(t,s) \left(\mathcal{L}\rho_1(s) - \frac{[e^{\lambda \mathcal{L}}-\mathcal{I}]}{\lambda} \rho_2(s)\right) \right\|_1 \\ && \nonumber = \left\| \int^{t}_{0}ds\, k(t,s) \left( \mathcal{L}\left(\rho_1(s)-\rho_2(s)\right) -\frac{1}{\lambda} \sum^{\infty}_{i=2} \frac{(\lambda \mathcal{L})^{i}}{i!} \rho_2 \right) \right\|_1 \nonumber \\ && \nonumber \leq \int^{t}_{0}ds\,  k(t,s) \|\mathcal{L}\| \| \rho_1(s)-\rho_2(s)\|_1 \\ &&+ \frac{\lambda^2 \|\mathcal{L}\|^2 e^{ \lambda \|\mathcal{L}\|}}{2}\int^{t}_{0}ds\, h(t,s) \|\rho_2(s)\|_1, \label{eq3}
\end{eqnarray}
where we have used the definitions $h(t,s)\equiv \int_s^td\tau\, H(\tau,s)$ and $k(t,s)\equiv \int_s^td\tau\, K(\tau,s)$. In Eq. \eqref{eq3}, we have used the triangle inequality and, then, the Lagrange bound for the Taylor series truncation on the last term, i.e., $\sum_{i=2}^\infty \frac{\lambda^i \| \mathcal{L}\| ^i}{i!}\leq \frac{\lambda^2 \|\mathcal{L}\|^2}{2}e^{\lambda \|\mathcal{L}\|}$. As in Proposition 1, we can now bound $\|\rho_2(s)\|_1$ by using the Gr\"onwall's inequality from Theorem~\ref{gro1}: $\|\rho_2(s)\|_1\leq e^{\int_0^sd\tau\, h(s,\tau)\left\|\mathcal{E}_\lambda-\mathcal{I}\right\|}$, with $\|\rho(0)\|_1=1$. At this point, we can bound the second term in Eq.~\eqref{eq3}, obtaining 
\begin{align}\label{eq4}
\|\rho_1(t)-\rho_2(t)\|\leq & \int^{t}_{0}ds\,  k(t,s) \|\mathcal{L}\| \| \rho_1(s)-\rho_2(s)\|_1 \nonumber \\
+ & \frac{\lambda^2 \|\mathcal{L}\|^2 e^{ \lambda \|\mathcal{L}\|}}{2\left\| \mathcal{E_\lambda}-\mathcal{I}\right\|}\left(e^{\int_0^tds\,h(t,s)\left\|\mathcal{E_\lambda}-\mathcal{I}\right\|}-1\right).
\end{align} 
Here, we have used that $h(s,\tau)\leq h(t, \tau)$, for $s\leq t$, and performed the integration. The second term in Eq.~\eqref{eq4} is positive and nondecreasing in time, so we can apply the Gr\"{o}nwall's inequality from Theorem~\ref{gro1},
\begin{align}\label{last}
\|\rho_1(t)-\rho_2(t)\|_1 \le \frac{\lambda^2 \|\mathcal{L}\|^2 e^{\lambda \|\mathcal{L}\|}}{2 \left\|\mathcal{E_\lambda}-\mathcal{I}\right\|} \left(e^{ \int^{t}_{0}ds\, h(t,s) \|\mathcal{E_\lambda}-\mathcal{I}\|} -1 \right) \times \nonumber \\
\times  e^{\int^{t}_{0}ds\, k(t,s) \|\mathcal{L}\|} \leq  \frac{\lambda \|\mathcal{L}\|^2 e^{\lambda \|\mathcal{L}\|}}{2 \left\|\mathcal{E_\lambda}-\mathcal{I}\right\|/\lambda} \left(e^{ ct \|\mathcal{E_\lambda}-\mathcal{I}\|/\lambda} -1 \right) e^{ct\|\mathcal{L}\|},
\end{align}
where we have set $h(t,s)=k(t,s)/\lambda$ and we have assumed $k(t,s)\leq c$. 

Finally, let us consider two parameter regimes:

(1) {\it First Regime}: For $c \| \mathcal{L}\| t\leq 1/e$, the expression in Eq.~\eqref{last} is bounded by $\varepsilon$, provided that $\lambda \leq \log\left(\frac{1}{c \left\|\mathcal{L}\right\|t}\right)\frac{\varepsilon}{\left\| \mathcal{L}\right\|}$,
\begin{align*}
\frac{\lambda \|\mathcal{L}\|^2 e^{\lambda \|\mathcal{L}\|}}{2} \frac{e^{ ct \|\mathcal{E_\lambda}-\mathcal{I}\|/\lambda} -1}{\left\|\mathcal{E_\lambda}-\mathcal{I}\right\|/\lambda} e^{ct\|\mathcal{L}\|} \leq \frac{e}{2}\lambda \left\| \mathcal{L}\right\| c \left\| \mathcal{L}\right\| t e^{\lambda \left\| \mathcal{L}\right\|} \\
\leq \frac{e}{2} \log\left(\frac{1}{c \left\| \mathcal{L}\right\| t}\right)\left(c \left\| \mathcal{L}\right\| t\right)^{1-\varepsilon}\varepsilon \leq \varepsilon , 
\end{align*}
where we have used in the first inequality that $z\equiv ct \|\mathcal{E_\lambda}-\mathcal{I}\|/\lambda\leq ct \left(e^{\lambda \left\| \mathcal{L}\right\|}-1\right) /\lambda\leq ct \left\| \mathcal{L}\right\| e^{\lambda\left\| \mathcal{L}\right\|}\leq\left( ct\left\| \mathcal{L}\right\| \right) ^{1-\varepsilon}<~e^{-1/2}$, in order to apply the inequality $(e^{z}-1)/z< e^{1/2}e^{e^{-1/2}}-1<e^{1/2}$ ($0\leq z<e^{-1/2}$), and the last inequality holds for $\varepsilon<1/2$.

(2) {\it Second Regime}: For $c\left\|\mathcal{L}\right\| t >1/e$, the expression in Eq.~\eqref{last} is bounded by $\varepsilon$, provided that $\lambda \leq \frac{e^{-(1+e^{\varepsilon})c\left\|\mathcal{L}\right\| t}}{c\left\|\mathcal{L}\right\|^2t }\varepsilon$. In fact, for this parameter choice, we have that $\lambda \left \| \mathcal{L}\right\| < \varepsilon$, which implies $\|\mathcal{E_\lambda}-\mathcal{I}\|/\lambda \leq \left(e^{\lambda \left\| \mathcal{L}\right\|}-1\right) /\lambda \leq e^{\varepsilon}\left\|\mathcal{L}\right\|$. Hence, the relation
\begin{align*}
\frac{\lambda \|\mathcal{L}\|^2 e^{\lambda \|\mathcal{L}\|}}{2} \frac{e^{ ct \|\mathcal{E_\lambda}-\mathcal{I}\|/\lambda} -1}{\left\|\mathcal{E_\lambda}-\mathcal{I}\right\|/\lambda} e^{ct\left\|\mathcal{L}\right\|} \leq \frac{\lambda \left(c\left\| \mathcal{L}\right\|^2t\right)  e^{\lambda \left\| \mathcal{L}\right\|}}{2}\times \\
\times e^{ct\left(\|\mathcal{E_\lambda}-\mathcal{I}\|/\lambda +\left\|\mathcal{L}\right\|\right)}\leq \frac{e^\varepsilon}{2}\varepsilon \leq \varepsilon
\end{align*}
holds. Here, we have used in the first inequality that $(e^{z}-1)/z<e^{z}$, applying it to $z\equiv ct \|\mathcal{E_\lambda}-\mathcal{I}\|/\lambda$, and the last inequality holds for $\varepsilon<1/2$.

\section{Observable decomposition in sum of unitary matrices}
Any observable $O$ can be decomposed as a sum of two unitary matrices $U_{a}$ and $U_b$, as $O=U_a + \gamma U_b$, with $\gamma > 0$ and $\|O\|\le 1+\gamma$ \cite{flu}. The first step is the diagonalization of $O$, $O=VDV^{\dag}$, and obtain the equations for $a_i$ and $b_i$, the eigenvalues of $U_a$ and $U_b$, as a function of $\gamma$ and $d_i$, the eigenvalues of $O$, as follows:
\begin{equation}
d_i=a_i + \gamma b_i, \qquad |a_i|=1, \qquad |b_i|=1.
\end{equation}
The eigenvalues are decomposed into real and imaginary parts,
\begin{eqnarray*}
&&\textrm{Re}(a_i)=\frac{d^2_i-\gamma_{-}}{2 d_i}, \qquad \textrm{Im}(a_i)=\frac{\sqrt{-d^4_i+2d^2_i\gamma_{+}-\gamma_{-}^2}}{2 d_i},  \\ 
&&\textrm{Re}(b_i)=\frac{d^2_i+\gamma_{-}}{2 d_i \gamma}, \qquad \textrm{Im}(b_i)=\frac{\sqrt{-d^4_i+2d^2_i\gamma_{+}-\gamma_{-}^2}}{2 d_i \gamma}, 
\end{eqnarray*}
with $\gamma_{\pm} = \gamma^2 \pm 1$, and the unitary matrices obtained,
\begin{equation}
\left( U_{a(b)} \right)_{i j}=V^{\dag}_{i n}a_n (b_n) V_{n j} .
\end{equation}

There is a restriction imposed by the fact that the imaginary parts of $a$ and $b$ have to be real numbers, which translates into the condition
\begin{equation}
|-1+d_i|\le\gamma\le 1 +d_i.
\end{equation}

\end{document}